\documentclass[aip,jmp,reprint,onecolumn]{revtex4-1} 

\usepackage{color}
\usepackage{graphicx}

\usepackage{hyperref}

\usepackage{bm,amssymb,amsmath,amsthm,mathrsfs,multirow,longtable,url,colonequals,verbatim,afterpage}



\theoremstyle{plain}
\newtheorem{theorem}{Theorem}[section]

\newtheorem{proposition}[theorem]{Proposition}
\newtheorem{conjecture}[theorem]{Conjecture}

\theoremstyle{definition}
\newtheorem{definition}[theorem]{Definition}

\theoremstyle{remark}

\DeclareSymbolFont{AMSb}{U}{msb}{m}{n}
\DeclareMathSymbol{\K}{\mathbin}{AMSb}{"4B}
\DeclareMathSymbol{\LL}{\mathbin}{AMSb}{"4C}
\DeclareMathSymbol{\N}{\mathbin}{AMSb}{"4E}
\DeclareMathSymbol{\Z}{\mathbin}{AMSb}{"5A}
\DeclareMathSymbol{\R}{\mathbin}{AMSb}{"52}
\DeclareMathSymbol{\Q}{\mathbin}{AMSb}{"51}
\DeclareMathSymbol{\I}{\mathbin}{AMSb}{"49}
\DeclareMathSymbol{\C}{\mathbin}{AMSb}{"43}
\DeclareMathSymbol{\F}{\mathbin}{AMSb}{"46}
\DeclareMathSymbol{\E}{\mathbin}{AMSb}{"45}

\def\ket#1{{|#1\rangle}}
\def\bra#1{{\langle#1|}}

\def\braket#1#2{{\langle#1|#2\rangle}}

\def\diag{{\operatorname{diag}}}
\def\tr{{\operatorname{tr}}}

\def\I{{\operatorname{I}}}

\begin{document}

\title{Fibonacci-Lucas SIC-POVMs}

\author{Markus Grassl}
\email{markus.grassl@mpl.mpg.de}
\affiliation{${}^1$Max-Planck-Institut f\"ur die Physik des Lichts, 91058
  Erlangen, Germany}

\author{Andrew J. Scott}
\email{dr.andrew.scott@gmail.com}
\noaffiliation

\begin{abstract}
We present a conjectured family of SIC-POVMs which have an additional
symmetry group whose size is growing with the dimension.  The symmetry
group is related to Fibonacci numbers, while the dimension is related
to Lucas numbers.  The conjecture is supported by exact solutions for
dimensions $d=4,8,19,48,124,323$, as well as a numerical solution for
dimension $d=844$.
\end{abstract}

\maketitle

\section{Introduction}
\label{sec:intro}

The study of equiangular complex lines in quantum information theory
started with the thesis of Zauner,\cite{Zau99,Zau11} and
independently by Renes et al.,\cite{RBSC04} who also coined the term
symmetric informationally-complete positive operator valued measure
(SIC-POVM).  Both conjectured that a SIC-POVM exists for all
dimensions and that it could be constructed as the orbit under the
Weyl-Heisenberg (WH) group. In addition, Zauner made the remark that
one could choose the so-called fiducial vector in a particular
eigenspace of an order-three unitary---this is now known as ``Zauner's
conjecture.'' In Ref.~\citenum{ScGr10}, we provided further support
for these conjectures via numerical solutions up to dimension $67$, as
well as a couple of exact solutions for various dimensions.  The
second author extended the list of numerical solutions up to dimension
$121$, confirming Zauner's conjecture.\cite{Sco17} Imposing additional
symmetries allowed to find solutions for certain dimensions up to
$323$.  Using the very program from the second author, the list of
numerical solutions was completed up to dimension $151$ in
Ref.~\citenum{FHS17}.

Appleby \cite{App05} studied the candidates for additional symmetries
of a SIC-POVM that is covariant with respect to the Weyl-Heisenberg
group.  He made the even stronger Conjecture~C that every fiducial
vector for such a SIC-POVM was an eigenvector of an order-three
unitary that is conjugate to Zauner's matrix.  This stronger version
of Zauner's conjecture turned out to be false,\cite{Gra05} but there
is a lot of support for Conjecture~A by Appleby.\cite{App05}  It
states that WH-covariant SIC-POVMs exist in all dimensions and that
the fiducial vector is an eigenvector of a canonical order-three
unitary.  All the solutions up do dimension $67$ presented in
Ref.~\citenum{ScGr10} possess such an additional order-three
symmetry---or even more---without imposing the symmetry.

In Ref.~\citenum{ScGr10}, we have also identified putative families of
additional symmetries, and those symmetries enabled us to find exact
solutions in dimensions as large as $48$.  The list of putative
symmetries was extended in Ref.~\citenum{Sco17} and helped to find
numerical solutions in dimensions $d=99, 111,
120,124,143,147,168,172,195,199,228,259$, and $323$ which, in addition
to being WH-covariant, have symmetries of order $6$ or $9$.

Analyzing the symmetry group of the numerical solutions,\cite{GrWa17}
it turned out that the solution in dimension $124$ did not only have
the imposed order-six symmetry, but a symmetry group of order
$30$. This discovery lead to the conjectured family of symmetries and
new solutions presented here.

\section{Background and notation}
\label{sec:background}
Throughout this article, we mainly use the notation from quantum
information.\cite{ScGr10} A SIC-POVM in dimension $d$ can be described
by a set of $d^2$ rank-one projectors $\Pi_i=\ket{\psi_i}\bra{\psi_i}$ on $\C^d$
such that
\begin{alignat}{9}
  \tr(\Pi_i\Pi_j)=\frac{1+\delta_{ij}d}{1+d}.
\end{alignat}
The Weyl-Heisenberg group in dimension $d$ is generated by the cyclic
shift operator $\hat{X}$ and its Fourier-transformed version $\hat{Z}$, given by
\begin{alignat}{9}
  \hat{X}=\sum_{i=0}^{d-1} \ket{i+1}\bra{i}
  \qquad\text{and}\qquad
  \hat{Z}=\sum_{j=0}^{d-1}\omega^j \ket{j}\bra{j},
\end{alignat}
where $\omega=\exp(\frac{2\pi i}{d})$ is a complex primitive $d$-th
root of unity and addition is modulo $d$. Ignoring phase factors, the
elements of the Weyl-Heisenberg group $X^aZ^b$ can be identified with
pairs $(a,b)\in\Z_d\times\Z_d$ of integers modulo
$d$. Appleby\cite{App05} makes a particular choice for the phases and
defines the displacement operators as
\begin{alignat}{9}
  \hat{D}_{(a,b)}=\tau^{ab}\hat{X}^a \hat{Z}^b,\label{eq:displacement}
\end{alignat}
where $\tau=-\exp(\frac{\pi i}{d})$. Then a Weyl-Heisenberg covariant
SIC-POVM is given by a fiducial vector $\ket{\psi_{(0,0)}}$ and the
vectors
\begin{alignat}{9}
  \ket{\psi_{(a,b)}}=\hat{D}_{(a,b)}\ket{\psi_{(0,0)}},
  \qquad\text{where $(a,b)\in\Z_d^2$.}
\end{alignat}
In all known cases, additional symmetries are elements of the extended
Clifford group.\cite{App05}  Again ignoring phase factors, the action
of a unitary symmetry on the displacement operators
\eqref{eq:displacement} can be described by an invertible $2\times 2$
matrix $F$ over $\Z_d$ with determinant $+1$.  Likewise, anti-unitary
symmetries correspond to matrices with determinant $-1$. (In order to
simplify the presentation here, we do not distinguish between odd and
even dimensions; a more rigorous discussion can be found in
Ref.~\citenum{App05}.)  A canonical order-three unitary\cite{App05} is
an element of the Clifford group for which the corresponding matrix
$F$ has trace $-1$ and order three.  The symmetry of Zauner's
conjecture corresponds to the matrix
\begin{alignat}{9}
  F_z=\begin{pmatrix}
  0& -1\\
  1&-1
  \end{pmatrix}.\label{eq:FibonacciMatrix}
\end{alignat}
For dimensions $d$ with $d\not\equiv 3 \bmod 9$, every canonical
order-three unitary is equivalent to Zauner's matrix.  For
$d=9\ell+3$, $\ell\ge 1$, however, there is additionally the matrix
\begin{alignat}{9}
  F_a=\begin{pmatrix}
  1& 3\\
  3\ell&-2
  \end{pmatrix}  
\end{alignat}
which corresponds to a canonical order-three unitary that is not
conjugate to Zauner's matrix.  In Refs.~\citenum{ScGr10} and
\citenum{Sco17}, several putative families of additional symmetries
for certain dimensions have been identified. Families of unitary
symmetries of order $2$ and $9$ are described by matrices $F_b$ and
$F_d$, respectively, anti-unitary symmetries of order $2$ and $6$ are
given by $F_c$ and $F_e$, respectively.  The latter can be
combined\cite{Sco17} to a family of order-six anti-unitary symmetries
$F_{e'}$.  While the sequence of dimensions for which such additional
symmetries might exist is unbounded, the order of the symmetry is at
most $9$.  Here we make the following conjecture.
\begin{conjecture}\label{conj:FibanacciLucasSICPOVM}
  For the infinite sequence of dimensions
  $d_k=\varphi^{2k}+\varphi^{-2k}+1$, where $\varphi=(1+\sqrt{5})/2$, $k\ge 1$,
  there exists a Weyl-Heisenberg covariant SIC-POVM that has an
  additional anti-unitary symmetry of order $6k$ given by the
  Fibonacci matrix
  \begin{alignat}{5}
    F_f=\begin{pmatrix}
    0&1\\
    1&1
    \end{pmatrix}.
  \end{alignat}
  We term such a set of vectors a \emph{Fibonacci-Lucas SIC-POVM}.
\end{conjecture}

The conjecture is supported by exact solutions for dimensions
$d=4,8,19,48,124,323$ as well as a numerical solution for $d_7=844$.

\section{The Fibonacci symmetry}
In this section, we study some properties of the additional symmetry
given by the Fibonacci matrix. We start with the definition of
Fibonacci and Lucas numbers.
\begin{definition}[Fibonacci numbers]
Let $F_n$ denote the $n$-th Fibonacci number given by the linear
recurrence relation $F_{n+1}=F_n+F_{n-1}$ for $n\ge 1$, with $F_0=0$,
$F_1=1$.
\end{definition}
Note that we follow the general convention to denote the $n$-th
Fibonacci number by $F_n$.  At the same time, we use the symbol $F$ to
denote the $2\times 2$ matrix over $\Z_d$ corresponding to the
symmetry of a WH-covariant SIC-POVM. It should always be clear from
the context whether $F$ refers to a Fibonacci number or a matrix.
\begin{definition}[Lucas numbers]
Let $L_n$ denote the $n$-th Lucas number given by the linear
recurrence relation $L_{n+1}=L_n+L_{n-1}$ for $n\ge 1$, with $L_0=2$,
$L_1=1$.
\end{definition}
Instead of the matrix $F_f$ over the integers modulo $d$, we consider
the matrix 
\begin{alignat}{5}
A=\begin{pmatrix}
0&1\\
1&1
\end{pmatrix}\in{\rm GL}(2,\Z)
\end{alignat}
with $\det(A)=-1$, over the integers. First note that the matrix $A$
factorizes as
\begin{alignat}{9}
  A=F_z\cdot J =
  \begin{pmatrix}
    0 & -1\\
    1 & -1
  \end{pmatrix}
  \begin{pmatrix}
    1 & 0\\
    0 & -1
  \end{pmatrix},\label{eq:Fibonacci_Zauner}
\end{alignat}
where $J$ is the matrix corresponding to complex conjugation in the
standard basis.

\begin{proposition}
For $n> 0$, the $n$-th power of $A$ is given by
\begin{alignat}{5}
A^{n}=
\begin{pmatrix}
F_{n-1}&F_{n}\\
F_{n}&F_{n+1}
\end{pmatrix}\label{eq:power_fibonacci_matrix}
\end{alignat}
\end{proposition}
\begin{proof}
The claim is true for $n=1$. By induction, we find
\begin{alignat}{5}
A^{n+1}=A^n A
=\begin{pmatrix}
F_{n-1}&F_{n}\\
F_{n}&F_{n+1}
\end{pmatrix}
\begin{pmatrix}
0&1\\
1&1
\end{pmatrix}
=\begin{pmatrix}
F_{n}&F_{n-1}+F_{n}\\
F_{n+1}&F_{n}+F_{n+1}
\end{pmatrix}
=\begin{pmatrix}
F_{n}&F_{n+1}\\
F_{n+1}&F_{n+2}.
\end{pmatrix}
\end{alignat}
\end{proof}

The dimension $d_k$ in Conjecture~\ref{conj:FibanacciLucasSICPOVM} is
given by $d_k=L_{2k}+1$ for $k\ge 1$ (see \eqref{eq:lucas_closed}).
The first ten dimensions in this sequence are $d=4,8, 19, 48, 124,
323, 844, 2208, 5779, 15128$.  We summarize some properties of the
sequence of dimensions, which is listed as sequence A065034 in the
On-Line Encyclopedia of Integer Sequences.\cite{OEIS} 
\begin{proposition}\label{prop:d_k}\ 
  \begin{enumerate}
  \item The sequence $d_k$ obeys the linear recurrence relation
    \begin{alignat}{9}
      d_{k+3}&{}=4d_{k+2}-4d_{k+1}+d_k.
    \end{alignat}
    \item Considered modulo $3$, the sequence $\tilde{d}_k=d_k\bmod 3$
      has period four and is given by $1,2,1,0,\,1,2,1,0,\ldots$ (for
      $k=1,2,3,\ldots$). This implies that $d_k$ is divisible by $3$
      if and only if $k=4\ell$. Then $d_{4\ell}\equiv 3 \bmod 9$.
    \item The square-free part of $(d_k+1)(d_k-3)$ equals $5$.
  \end{enumerate}
\end{proposition}\pagebreak{3}
\begin{proof}\ 
  \begin{enumerate}
    \item From the recurrence relation for the Lucas numbers, we obtain
      $L_{n+2}= 3 L_n-L_{n-2}$. Together with the definition
      $d_k=L_{2k}+1$, we first compute
      \begin{alignat}{9}
        d_{k+2}=L_{2k+4}+1&{}=3L_{2k+2}-L_{2k}+1=3 d_{k+1}-d_{k}-1.\label{eq:d_K1}
      \end{alignat}
      Then, using \eqref{eq:d_K1} twice, we get
      \begin{alignat}{9}
        d_{k+3}=3 d_{k+2}-d_{k+1}-1&{}=4 d_{k+2}-d_{k+2}-d_{k+1}-1\nonumber\\
        &{}=4 d_{k+2}-(3 d_{k+1}-d_{k}-1)-d_{k+1}-1\nonumber\\
        &{}=4 d_{k+2}-4 d_{k+1}+d_{k}.
      \end{alignat}
    \item Consider the generating function $D(z)=\sum_{k\ge 0} d_k
      z^k$. Using standard techniques\cite{GKP94} one can show that
      \begin{alignat}{5}
        D(z)=\frac{-4z^2 + 8z - 3}{z^3 - 4z^2 + 4z - 1}
        &{}\equiv\frac{-z^2 - z}{z^3 - z^2 + z - 1}\pmod{3}\nonumber\\
        &{}=\frac{z^3+2z^2 + z}{1-z^4}=(z^3+2z^2 + z)(1+z^4+z^8+\ldots).
      \end{alignat}
      This shows that $\tilde{d}_k$ has period $4$, and that
      $d_k\equiv 0 \bmod 3$
      if and only if $k=4\ell$.  For $k=4\ell$, we compute
      \begin{alignat}{5}
        d_{4\ell}=L_{8\ell}+1=L^2_{4\ell}-1=5 F_{4\ell}^2+3,
      \end{alignat}
      where we have used \eqref{eq:L_4n} and
      \eqref{eq:squarefree}. From \eqref{eq:divisibility_Fibonacci} it
      follows that $F_{4l}\equiv 0 \bmod F_4$. Together with $F_4=3$,
      this implies that $d_{4\ell}\equiv 3 \bmod 9$.
    \item We have
      \begin{alignat}{9}
        (d_k+1)(d_k-3)=(d_k-1)^2-4=L_{2k}^2-4=5F_{2k}^2,
      \end{alignat}
      where the last equality follows from \eqref{eq:squarefree}.
  \end{enumerate}
\end{proof}
The last statement that the squarefree part of $(d_k+1)(d_k-3)$ equals
$D=5$ has to be seen in the context of the conjecture\cite{AFMY16}
that the number field of smallest degree containing the fiducial
projector in the standard basis is a ray class field over
$\Q(\sqrt{D})$. It turned out that the conjecture is true for all our
exact solutions.

Next we show that the matrix $F_f$ determining the additional symmetry
for dimension $d_k$ has the desired properties.
\begin{proposition}
When considered as an element $\tilde{A}\in{\rm GL}(2,\Z_d)$ with
$d=d_k=L_{2k}+1$, the matrix $F_f=\tilde{A}$ has the following properties:
\begin{enumerate}
\item $F_f^{2k}$ has trace $-1$.
\item For $k$ even, $F_f^{3k}$ is a scalar multiple of identity.
\item $F_f$ has order $6k$.
\item For $k\equiv 0\bmod 4$, $F_f^{2k}$ is conjugate to
  $F_a$; otherwise, it is conjugate to Zauner's matrix $F_z$.
\end{enumerate}
\end{proposition}
\begin{proof}
  From \eqref{eq:power_fibonacci_matrix} it follows that the entries
  of the matrix $A^{2k}$ over $\Z$ are strongly monotonic increasing
  in $k$.  We have
  \begin{alignat}{5}
    A^{2k}=
    \begin{pmatrix}
      F_{2k-1}&F_{2k}\\
      F_{2k}&F_{2k+1}
    \end{pmatrix}.
  \end{alignat}
  From \eqref{eq:lucas_sum_fibonacci} it follows that $d_k=L_{2k}+1>
  F_{2k+1}>F_{2k}$, and hence the order of $\tilde{A}\in{\rm GL}(2,\Z_{d_k})$ is strictly
  larger than $2k$. The determinant is $\det(A^{2k})=1$, and
  \begin{alignat}{5}
    \tr(A^{2k})=F_{2k-1}+F_{2k+1}=L_{2k}\equiv -1 \bmod d_k,
  \end{alignat}
  where we have again used \eqref{eq:lucas_sum_fibonacci}.  It follows
  that the characteristic polynomial of $F_f^{2k}$ is $x^2+x+1$, and
  $F_f^{2k}$ corresponds to a canonical order-three unitary.\cite{App05}
  As $F_f^{2k}$ has order three, the order of $F_f$ is
  a divisor of $6k$.  We already know that the order is strictly
  larger than $2k$, and therefore we are left with the possibilities
  $3k$ and $6k$.

  When $k$ is odd, the determinant of $A^{3k}$ equals $-1$, and
  therefore the order of $F_f$ is $6k$.  Now assume that $k=2\ell$
  is even. Then
  \begin{alignat}{5}
    A^{3k}=A^{6\ell}
    =\begin{pmatrix}
    F_{6\ell-1}&F_{6\ell}\\
    F_{6\ell}&F_{6\ell+1}
    \end{pmatrix}
    =\begin{pmatrix}
    F_{6\ell-1}&F_{6\ell}\\
    F_{6\ell}&F_{6\ell-1}+F_{6\ell}
    \end{pmatrix}.
  \end{alignat}
  From \eqref{eq:F_6n}, we have $F_{6\ell}\equiv 0 \bmod (L_{4\ell}+1)$,
  and hence
  \begin{alignat}{5}
    \tilde{A}^{3k}
    =\begin{pmatrix}
    F_{6\ell-1}&0\\
    0&F_{6\ell-1}
    \end{pmatrix}.\label{eq:A_3k}
  \end{alignat}
  From \eqref{eq:F_6n-1a}, we have
  $F_{6\ell-1}+L_{2\ell}=F_{2\ell-1}(L_{4\ell}+1)$, which implies
  $F_{6\ell-1}\equiv -L_{2\ell}\bmod (L_{4\ell}+1)$. In particular, we
  have $F_{6\ell-1}\not\equiv 1\bmod (L_{4\ell}+1)$, i.\,e.,
  $\tilde{A}^{3k}$ is different from identity, and hence the order of
  $F_f$ is $6k$ in this case as well.  Note that the matrix
  \eqref{eq:A_3k} is related to the matrix $F_b$ in
  Ref.~\citenum{ScGr10}.  From \eqref{eq:L_4n} it follows that the
  dimension $d_{2\ell}=L_{4\ell}+1$ is $L_{2\ell}^2-1$, i.\,e., a square
  minus one.

  In order to prove the final part, first note that the dimension
  $d_k$ is co-prime to three when $k$ is not a multiple of four, and
  for those dimensions  any matrix of order three with trace $-1$ is
  conjugate to Zauner's matrix. For $k=4\ell$, we have
  \begin{alignat}{5}
    A^{2k}=A^{8\ell}=
    \begin{pmatrix}
      F_{8\ell-1}&F_{8\ell}\\
      F_{8\ell}&F_{8\ell+1}
    \end{pmatrix}
    =\begin{pmatrix}
      F_{8\ell-1}&F_{8\ell}\\
      F_{8\ell}&F_{8\ell-1}+F_{8\ell}
    \end{pmatrix}.
  \end{alignat}
  Again using \eqref{eq:divisibility_Fibonacci}, $F_{8\ell}\equiv
  0\bmod F_8$, and with $F_8=21$, it follows that $A^{2k}=A^{8l}\bmod 3$ is
  proportional to identity. Furthermore, as $\tilde{A}^{2k}$ has order
  three, it follows that $A^{2k}=I\bmod 3$. Therefore,
  $\tilde{A}^{2k}$ is not conjugate to Zauner's matrix, but to $F_a$.

\end{proof}
In summary, the matrix $F_f$ is a candidate for an anti-unitary
symmetry of order $6k$ in dimension $d_k=L_{2k}+1$.  In the following
we show that exact fiducial vectors with this additional symmetry
exist for the first six dimensions up to $d_6=323$.  We have a
numerical solution for the next dimension $d_7=844$ as well.

\section{Exact solutions}
We can factorize $F_f$ as
\begin{alignat}{9}
  F_f=J F_f'=
  \begin{pmatrix}
    1 & 0\\
    0 & -1
  \end{pmatrix}
  \begin{pmatrix}
    0 & 1\\
  -1 & -1
  \end{pmatrix}
\end{alignat}
and obtain the corresponding unitary for the latter ${\rm SL}(2,\Z_d)$ matrix as\cite{App05}
\begin{alignat}{9}
  \hat{U}_{F_f'}=\frac{1}{\sqrt{d}}\sum_{r,s=0}^{d-1}\tau^{-r^2-2rs}\ket{r}\bra{s},
\end{alignat}
where $\tau=-\exp(\pi i/d)$.  Fiducial vectors $\ket{\psi}$ with the
$F_f$ symmetry are con-eigenvectors of this unitary, i.\,e.
\begin{alignat}{5}
  \hat{J} \hat{U}_{F_f'} \ket{\psi} = e^{i\phi} \ket{\psi}
\end{alignat}
for some irrelevant phase $\phi$, where $\hat{J}$ denotes the anti-linear operator corresponding to complex
conjugation in the standard basis.

A rigorous description of the extended Clifford group in even dimensions in fact requires  
$2\times 2$ matrices over $\Z_{2d}$ rather than $\Z_d$, a complication that has been avoided up until now.
To relate the following solutions to previous work, however, we now quote matrices over $\Z_{2d}$ as appropriate 
and follow Appleby's correspondence  precisely (Theorem 2 of Ref.~\citenum{App05}). The quoted symmetry 
group of $2\times 2$ matrices over $\Z_{2d}$ then doubly covers the corresponding group of extended Clifford operators 
(modulo its center).

\subsection{Dimension $d=4$}
In Ref.~\citenum{ScGr10}, the numerical solution labeled $4a$ has a
symmetry group of order $6$ generated by
\begin{alignat}{9}
  F_c F_z=\begin{pmatrix}
  1 & 2\\
  6 & 3
  \end{pmatrix}
  \begin{pmatrix}
    0 & 3\\
    5 & 3
  \end{pmatrix}=
  \begin{pmatrix}
    2 & 1\\
    7 & 3
  \end{pmatrix}.
\end{alignat}
The $7$-th power of this matrix is conjugate to $F_f$ by 
\begin{alignat}{9}
G=
  \begin{pmatrix}
  1 & 3\\
  3 & 2
  \end{pmatrix},
\end{alignat}
 as
\begin{alignat}{9}
G( F_c F_z)^7G^{-1}=
  \begin{pmatrix}
  1 & 3\\
  3 & 2
  \end{pmatrix}
  \begin{pmatrix}
    2 & 1\\
    7 & 3
  \end{pmatrix}^7
  \begin{pmatrix}
  1 & 3\\
  3 & 2
  \end{pmatrix}^{-1}=
  \begin{pmatrix}
  1 & 3\\
  3 & 2
  \end{pmatrix}
  \begin{pmatrix}
    2 & 5\\
    3 & 7
  \end{pmatrix}
  \begin{pmatrix}
  2 & 5\\
  5 & 1
  \end{pmatrix}=
  \begin{pmatrix}
    0 & 1\\
    1 & 1
  \end{pmatrix}=F_f.
\end{alignat}
A non-normalized exact fiducial vector with the symmetry $F_f$ is given by
\begin{alignat}{9}
\ket{\psi^{4a}}=
\text{\small$\begin{pmatrix}
8\sqrt{2}-8\\
\left((\sqrt{10}+\sqrt{2})\sqrt{1+\sqrt{5}}+4\sqrt{2}-4\right)-\left((-\sqrt{10}-\sqrt{2}+2\sqrt{5}+2)\sqrt{1+\sqrt{5}}+4\right)I\\
8I\\
\left(-(\sqrt{10}+\sqrt{2})\sqrt{1+\sqrt{5}}+4\sqrt{2}-4\right)-\left((\sqrt{10}+\sqrt{2}-2\sqrt{5}-2)\sqrt{1+\sqrt{5}}+4\right)I\\
\end{pmatrix}$},
\end{alignat}
where $I^2=-1$. This is the numerical solution $4a$ 
translated by the unitary corresponding to $G$.
The number field containing a fiducial projector is
$\E^{4a}=\Q(\sqrt{5},\sqrt{2},\sqrt{1+\sqrt{5}},\sqrt{-1})$, which is an
Abelian extension of $\Q(\sqrt{5})$.

\subsection{Dimension $d=8$}
In Ref.~\citenum{ScGr10}, the numerical solution labeled $8b$ has a
symmetry group of order $12$ generated by
\begin{alignat}{9}
F=
  \begin{pmatrix}
    6 & 11\\
    5 & 1
  \end{pmatrix}.
\end{alignat}
The $11$-th power of this matrix is conjugate to $F_f$ by 
\begin{alignat}{9}
G=
  \begin{pmatrix}
  5 & 5\\
  4 & 1
  \end{pmatrix},
\end{alignat}
 as
\begin{alignat}{9}
GF^{11}G^{-1}=
  \begin{pmatrix}
    5 & 5\\
    4 & 1
  \end{pmatrix}
  \begin{pmatrix}
    6 & 11\\
    5 & 1
  \end{pmatrix}^{11}
  \begin{pmatrix}
    5 & 5\\
    4 & 1
  \end{pmatrix}^{-1}=
  \begin{pmatrix}
    5 & 5\\
    4 & 1
  \end{pmatrix}
  \begin{pmatrix}
    7 & 3\\
    13 & 10
  \end{pmatrix}
  \begin{pmatrix}
    1 & 11\\
    12 & 5
  \end{pmatrix}=
  \begin{pmatrix}
    0 & 1\\
    1 & 1
  \end{pmatrix}
  =F_f.
\end{alignat}
A non-normalized exact fiducial vector with the symmetry $F_f$ is given by
\begin{alignat}{9}
  \ket{\psi^{8b}}=
\text{\footnotesize$\begin{pmatrix}
-4\sqrt{5}+12\\
\bigl((2\sqrt{2}-2)s_1s_2+(-\sqrt{10}+3\sqrt{2}+2\sqrt{5}-6)s_1-4s_2\bigr)I+2s_1s_2+(\sqrt{10}-3\sqrt{2})s_1-2\sqrt{10}+6\sqrt{2}\\
-4\sqrt{2}s_2I\\
\bigl((2\sqrt{2}-2)s_1s_2+(\sqrt{10}-3\sqrt{2}-2\sqrt{5}+6)s_1+4s_2\bigr)I+2s_1s_2+(-\sqrt{10}+3\sqrt{2})s_1-2\sqrt{10}+6\sqrt{2}\\
4\sqrt{5}-12\\
\bigl((-2\sqrt{2}+2)s_1s_2+(\sqrt{10}-3\sqrt{2}-2\sqrt{5}+6)s_1-4s_2\bigr)I-2s_1s_2+(-\sqrt{10}+3\sqrt{2})s_1-2\sqrt{10}+6\sqrt{2}\\
4\sqrt{2}s_2I\\
\bigl((-2\sqrt{2}+2)s_1s_2+(-\sqrt{10}+3\sqrt{2}+2\sqrt{5}-6)s_1+4s_2\bigr)I-2s_1s_2+(\sqrt{10}-3\sqrt{2})s_1-2\sqrt{10}+6\sqrt{2}
\end{pmatrix}$},
\end{alignat}
where $s_1=\sqrt{2+\sqrt{2}}$ and $s_2=\sqrt{\sqrt{5}-1}$. This is the numerical solution $8b$ 
translated by the unitary corresponding to $G$. The number field containing the fiducial projector is an Abelian
extension of $\Q(\sqrt{5})$ given by
\begin{alignat}{9}
  \E^{8b}=\Q(\sqrt{2},\sqrt{5},\sqrt{\sqrt{5}-1},\sqrt{2+\sqrt{2}},\sqrt{-1}).
\end{alignat}
The exact solution $8b$ in Ref.~\citenum{ScGr10} has a symmetry group generated by
\begin{alignat}{9}
F'=
  \begin{pmatrix}
  1 & 5\\
  13 & 0
  \end{pmatrix},
\end{alignat}
which is again conjugate to $F_f$, i.e., $HF'H^{-1}=F_f$ for the choice
\begin{alignat}{9}
H=
  \begin{pmatrix}
  1 & 4\\
  5 & 5
  \end{pmatrix}.
\end{alignat}
The unitary corresponding to $H$ translates the exact solution in Ref.~\citenum{ScGr10} 
to that given above.

\subsection{Dimension $d=19$}
In Ref.~\citenum{ScGr10}, the numerical solution labeled $19e$ has a
symmetry group of order $18$ generated by
\begin{alignat}{9}
F=
  \begin{pmatrix}
    3 & 12\\
    7 & 15
  \end{pmatrix}.
\end{alignat}
The $17$-th power of this matrix is conjugate to $F_f$ by 
\begin{alignat}{9}
G=
  \begin{pmatrix}
  11 & 10\\
  0 & 7
  \end{pmatrix},
\end{alignat}
 as
\begin{alignat}{9}
GF^{17}G^{-1}=
  \begin{pmatrix}
    11 & 10\\
    0 & 7
  \end{pmatrix}
  \begin{pmatrix}
    3 & 12\\
    7 & 15
  \end{pmatrix}^{17}
  \begin{pmatrix}
    11 & 10\\
    0 & 7
  \end{pmatrix}^{-1}=
  \begin{pmatrix}
    11 & 10\\
    0 & 7
  \end{pmatrix}
  \begin{pmatrix}
    4 & 12\\
    7 & 16
  \end{pmatrix}
  \begin{pmatrix}
    7 & 9\\
    0 & 11
  \end{pmatrix}=
  \begin{pmatrix}
    0 & 1\\
    1 & 1
  \end{pmatrix}
  =F_f.
\end{alignat}
A non-normalized exact fiducial vector with the symmetry $F_f$ can be
found online.\cite{SICs_online} This is the numerical solution $19e$
translated by the unitary corresponding to $G$.  The exact solution
$19e$ in Ref.~\citenum{ScGr10} has a symmetry group generated by the
diagonal matrix $F'=\diag(15,5)$, which yields a unitary symmetry that
is just a permutation.  Hence, the corresponding fiducial vector has a
very compact representation. For
\begin{alignat}{9}
H=
  \begin{pmatrix}
  8 & 5\\
  6 & 6
  \end{pmatrix}
\end{alignat}
we have $HF'H^{-1}=F_f$,
i.\,e., the symmetry is conjugate to $F_f$ as well. The unitary corresponding to $H$
also translates the exact solution in Ref.~\citenum{ScGr10} 
to the alternative given online.\cite{SICs_online}

\subsection{Dimension $d=48$}
The numerical solution in Ref.~\citenum{ScGr10} given for
orbit $48g$ has a symmetry group of order $24$ generated by
\begin{alignat}{9}
F=
 \begin{pmatrix}
    4 & 37\\
    25 & 63
  \end{pmatrix}.
\end{alignat}
The $41$-st power of this matrix is conjugate to $F_f$ by 
\begin{alignat}{9}
G=
  \begin{pmatrix}
  10 & 47\\
  21 & 22
  \end{pmatrix},
\end{alignat}
 as
\begin{alignat}{9}
GF^{41}G^{-1}&{}=
  \begin{pmatrix}
  10 & 47\\
  21 & 22
  \end{pmatrix}
  \begin{pmatrix}
   4 & 37\\
    25 & 63
  \end{pmatrix}^{41}
  \begin{pmatrix}
   10 & 47\\
  21 & 22
  \end{pmatrix}^{-1}\nonumber\\
  &{}=
  \begin{pmatrix}
   10 & 47\\
  21 & 22
  \end{pmatrix}
  \begin{pmatrix}
    61 & 25\\
    61 & 36
  \end{pmatrix}
  \begin{pmatrix}
    22 & 49\\
    75 & 10
  \end{pmatrix}=
  \begin{pmatrix}
    0 & 1\\
    1 & 1
  \end{pmatrix}
  =F_f.
\end{alignat}
A non-normalized exact fiducial vector with the symmetry $F_f$ can be
found online.\cite{SICs_online}
This is the numerical solution $48g$ translated by the unitary corresponding to $G$. 
The exact solution $48g$ in Ref.~\citenum{ScGr10} also has the symmetry group generated by $F$ once it is displaced by $\hat{D}_{(41,15)}$.
The unitary corresponding to
\begin{alignat}{9}
H=
  \begin{pmatrix}
  5 & 16\\
  33 & 29
  \end{pmatrix}
\end{alignat}
then translates this exact solution to the alternative given online.\cite{SICs_online}

\subsection{Dimension $d=124$}
The numerical solution of Ref.~\citenum{Sco17} labeled $124a$, which
was obtained by imposing the symmetries $F_z$ and $F_b$ of combined
order $6$, turned out to have a symmetry of order $30$ generated by
\begin{alignat}{9}
F=
  \begin{pmatrix}
     58 & 133\\
    115 & 191
  \end{pmatrix}.
\end{alignat}
This symmetry is conjugate to $F_f$ by 
\begin{alignat}{9}
G=
  \begin{pmatrix}
  100 & 15\\
  85 & 45
  \end{pmatrix},
\end{alignat}
 as
\begin{alignat}{9}
GFG^{-1}&{}=
  \begin{pmatrix}
    100 & 15\\
  85 & 45
  \end{pmatrix}
  \begin{pmatrix}
   58 & 133\\
    115 & 191
  \end{pmatrix}
  \begin{pmatrix}
   100 & 15\\
  85 & 45
  \end{pmatrix}^{-1}\nonumber\\
  &{}=
  \begin{pmatrix}
  100 & 15\\
  85 & 45
  \end{pmatrix}
  \begin{pmatrix}
    58 & 133\\
    115 & 191
  \end{pmatrix}
  \begin{pmatrix}
    45 & 233\\
     163 & 100
  \end{pmatrix}=
  \begin{pmatrix}
    0 & 1\\
    1 & 1
  \end{pmatrix}
=F_f.
\end{alignat}
In order to find an exact fiducial vector with the anti-unitary
symmetry $F_f$, we use the factorization \eqref{eq:Fibonacci_Zauner}
of $F_f$ into Zauner's matrix $F_z$ and complex conjugation. We
identify $\C^{124}$ with the real space $\R^{248}$ and find a
$10$-dimensional eigenspace of the corresponding $\R$-linear
transformation, i.\,e., we have to solve for $10$ real variables.  As
$124=2^2\times 31$, we can apply a global change of basis in order to
obtain more sparse vectors.\cite{ABBEGL04} It turned out that it
sufficient to consider only those vectors of the SIC-POVM that are
obtained as the orbit of the fiducial vector with respect to cyclic
shift operation $\hat{X}$, i.\,e., the equations
\begin{alignat}{9}
  \braket{\psi_{(0,0)}}{\psi_{(0,0)}}&{}=1,\\
  \text{and}\qquad
  |\bra{\psi_{(0,0)}}\hat{X}^a\ket{\psi_{(0,0)}}|^2&{}=\frac{1}{125}\qquad\text{for $a=1,\ldots,123$.}\label{eq:shift124}
\end{alignat}
Moreover, the equations for the exponents $a$ and $-a$ in
\eqref{eq:shift124} are identical, so that we have only $63$ equations
in total.  Computing a Gr\"obner basis with the modular approach for
polynomial equations over number fields implemented in
Magma\cite{Magma} V2.22-6 on a system with an Intel Xeon X5680 processor
with 3.33\, GHz clock speed took about 27 CPU hours and 33 GB of RAM.
It has been conjectured that the minimal field containing a fiducial
projector is the ray class field $\E$ over $\Q(\sqrt{5})$ with
conductor $d'=2d=248$ and ramification at both infinite
places.\cite{AFMY16} Using Magma,\cite{Magma} we computed a slightly
optimized representation of $\E$ as
\begin{alignat}{9}
\E^{124a}=\Q(\sqrt{5},\sqrt{2},\sqrt{31},\sqrt{-1},\sqrt{(\sqrt{5}-1)/2},\sqrt{6+\sqrt{5}},s_1,s_2,s_3),
\end{alignat}
where $s_1$ is a root of the polynomial $f_1(x)=x^3+x^2-10x-8$, $s_2$
is a root of the polynomial $f_2(x)=x^5-x^4-12x^3+21x^2+x-5$, and
$s_3$ is a root of the polynomial $f_3(x)=x^3+x^2-41x+85$.  These
three polynomials have only real roots, i.\,e., the only complex
number among the generators of $\E^{124a}$ is the imaginary unit $\sqrt{-1}$.
The generators of the Galois group of $\E^{124a}/\Q$ are given as follows
(only the non-trivial action on the generators is listed):
\begin{small}
\begin{alignat}{9}
  g_1\colon && \sqrt{5}&{}\mapsto -\sqrt{5};&
    \quad\sqrt{(\sqrt{5}-1)/2}&{}\mapsto\frac{(\sqrt{5}-1)\sqrt{-1}}{2}\sqrt{(\sqrt{5}-1)/2};\nonumber\\
    &&&&\quad\sqrt{6+\sqrt{5}}&{}\mapsto\frac{(\sqrt{5}-6)\sqrt{31}}{31}\sqrt{6+\sqrt{5}}  \nonumber\\
  g_2\colon && \sqrt{2}&{}\mapsto -\sqrt{2}\nonumber\\
  g_3\colon && \sqrt{31}&{}\mapsto -\sqrt{31}\nonumber\\
  g_4\colon && \sqrt{-1}&{}\mapsto -\sqrt{-1}\nonumber\\
  g_5\colon && \sqrt{(\sqrt{5}-1)/2}&{}\mapsto -\sqrt{(\sqrt{5}-1)/2}\kern-3cm\nonumber\\
  g_6\colon && \sqrt{6+\sqrt{5}}&{}\mapsto -\sqrt{6+\sqrt{5}}\kern-3cm\nonumber\\
  g_7\colon && s_1&{}\mapsto (s_1^2-s_1-8)/2\kern-3cm\nonumber\\
  g_8\colon && s_2&{}\mapsto \left(2s_2^4-s_2^3-22s_2^2+31s_2\right)/5\kern-3cm\nonumber\\
  g_9\colon && s_3&{}\mapsto \left(4\sqrt{5}s_3^2 + (17\sqrt{5} - 5)s_3 - 105\sqrt{5} - 5\right)/10\kern-5cm
\end{alignat}
\end{small}%
The Galois group has order $2880$ and is isomorphic to $C_{30} \times
\bigl((C_6 \times C_2 \times C_2 \times C_2) \rtimes C_2\bigr)$. It
turns out that a fiducial vector can indeed be found in the ray class
field. The most complex step in the computation, however, was to find
a representation of the fiducial vector in $\E^{124a}$.  For this, we
had to factorize several polynomials over $\E^{124a}$.  A
non-normalized exact fiducial vector can be found
online.\cite{SICs_online} This is the numerical solution $124a$ of 
Ref.~\citenum{Sco17} translated by the unitary corresponding to $G$ above.

\subsection{Dimension $d=323$}
In order to find an exact fiducial vector with anti-unitary symmetry
$F_f$, we started as in the case for dimension $d=124$. Identifying
$\C^{323}$ with $\R^{646}$, we find an eigenspace of real dimension
$19$, i.e., we have to solve for $19$ real variables. Using the
factorization $d=17\times 19$ and $19\bmod 3=1$, we can apply a change
of basis by a Clifford transformation\cite{App05} such that the additional
symmetry is diagonal modulo $19$. Thereby we obtain not only a more
sparse eigenbasis, but at the same time the basis can be represented
in $\Q(\zeta_{17})$ (the cyclotomic field generated by the primitive
$17$-th root of unity $\zeta_{17}=\exp(2\pi i/17)$) instead of
$\Q(\zeta_{323})$.  As the dimension $d=323=18^2-1$, we can
additionally use  observations from Ref.~\citenum{ABDF17} concerning the
phases $\phi_{(a,b)}$ in the overlap
\begin{alignat}{5}
  \bra{\psi_{(0,0)}}\hat{D}_{{a,b}}\ket{\psi_{(0,0)}}=\frac{e^{i\phi_{(a,b)}}}{\sqrt{d+1}}.
\end{alignat}
When the fiducial vector ${\psi_{(0,0)}}$ is properly ``aligned'', the
phases are zero when $a$ and $b$ are both multiples of $19$. This
yields $17^2=289$ polynomial equations of degree two.  Furthermore,
the phases are related to those of a SIC-POVM in dimension $19=323/17$
when $a$ and $b$ are both multiples of $17$.  The possible values are
obtained from the exact solution $19e$ given above.  Using the new
numerical solution $323c$ and floating-point approximations of the
exact phases, we are able to determine those phases and obtain another
$19^2-1=360$ polynomial equations of degree two.  It should be noted
that after linear reduction of these equations, the $19$-th root of
unity is no longer needed.  While we have in total $649$ equations for
only $19$ variables, it turns out that the corresponding variety has
dimension $4$, i.e., these equations specify the fiducial vector only
up to four free parameters. Adding the equations for the cyclic shifts
of the fiducial vector, i.e.,
\begin{alignat}{5}
  |\bra{\psi_{(0,0)}}\hat{X}^a\ket{\psi_{(0,0)}}|^2&{}=\frac{1}{324}\qquad\text{for $a=1,\ldots,322$}\label{eq:shift323}
\end{alignat}
of degree $4$, we are left with a zero-dimensional variety, i.e., a
finite number of candidates for the fiducial vector.  Like in the case
for dimension $d=124$, the most time-consuming step is not the
computation of a Gr\"obner basis for the system of polynomial
equations, but to find the exact roots of the polynomials in the
corresponding number field and to identify which of the roots are
real.

The ray class field $\E^{323c}$ over $\Q(\sqrt{5})$ has degree
$10\,368$ over $\Q$. It contains a primitive $323$-th root of unity
$\zeta_{323}$ (and hence also both $\zeta_{17}$ and $\zeta_{19}$), and
can be generated as
\begin{alignat}{5}
  \E^{323c}=\Q(\sqrt{5},\sqrt{-(1+2\sqrt{5})},\tau,\zeta_{17},\zeta_{19}),
\end{alignat}
where $\tau$ is the root of a polynomial of degree $9$ over
$\Q(\sqrt{5})$. The Galois group is isomorphic to
$C_{144}\times\bigl((C_{18}\times C_2)\rtimes C_2\bigr)$.
Interestingly, the fiducial projector
$\Pi_{(0,0)}=\ket{\psi_{(0,0)}}\bra{\psi_{(0,0)}}$ can be expressed in
a subfield
\begin{alignat}{5}
\E_\text{fid}^{323c}=\Q(\sqrt{5},\sqrt{-(1+2\sqrt{5})},\tau,\zeta_{17})
<\E^{323c}=\E_\text{fid}^{323c}(\zeta_{19})
\end{alignat}
of degree $576$ over $\Q$.  A non-normalized exact fiducial vector as
well as more details on the representation of the ray class field and
its Galois group can be found online.\cite{SICs_online} The
anti-unitary symmetry group of the corresponding SIC-POVM has been
verified to be conjugate to the group of order $36$ generated by the
matrix $F_f$.

While the fiducial projector together with the displacement operators
generate the ray class field $\E^{323c}$, as conjectured in
Ref. \citenum{AFMY16}, the Clifford orbit of the SIC-POVM contains a
projector in a subfield of considerably smaller degree.  Whether the
field $\E_\text{fid}^{323c}$ has the smallest possible degree among
the fields generated by a projector in the Clifford orbit of a
SIC-POVM in dimension $323$ is open, as well as the question what the
lowest possible degree for arbitrary dimensions is.

\section{Numerical solutions}

Our numerical search for solutions followed Refs.~\citenum{ScGr10} and
\citenum{Sco17} by first translating the Welch bound on a set of unit
vectors in $\C^d$ to the case of a SIC-POVM: for any\/
$\ket{\phi}\in\C^d$,
\begin{align}\label{eq:sicbound}
\frac{1}{d^3}\sum_{a,b,a',b'}|\bra{\phi}{\hat{D}_{(a',b')}}^\dag \hat{D}_{(a,b)}\ket{\phi}|^4=\sum_{j,k}\Big|\sum_l\braket{\phi}{j+l}\braket{l}{\phi}\braket{\phi}{k+l}\braket{j+k+l}{\phi}\Big|^2 &\geq \frac{2}{d+1}\:,
\end{align}
with equality if and only if $\ket{\phi}$ is a fiducial vector for a
WH covariant SIC-POVM.  The condition for equality
follows\cite{RBSC04} from the 2-design property of a SIC-POVM.  We may
now search for fiducial vectors by simply minimizing the LHS of the
inequality in Eq.~(\ref{eq:sicbound}), parameterized as a function of
the real and imaginary parts of the $d$ complex numbers
$\braket{j}{\phi}$, until the bound on the RHS is met.

To search for Fibonacci-Lucas SIC-POVMs we applied the method of
Ref.~\citenum{Sco17} for general anti-unitary symmetries.  Suppose
that $\hat{J}\hat{U}\ket{\psi}$ = $\lambda\ket{\psi}$ for some unitary $\hat{U}$
with
$\big(\hat{J}\hat{U}\big)^{2n}=\big(\hat{J}\hat{U}\hat{J}\hat{U}\big)^n=\big(\overline{U}\hat{U}\big)^n=I$,
where $\overline{U}$ denotes complex conjugation of matrix components
in the standard basis. We must have $|\lambda|=1$ and may in fact
assume $\lambda=1$ for con-eigenvalues. Now given that the projector
onto the eigenspace of the unitary $\overline{U}\hat{U}$ with eigenvalue $1$
is
\begin{align}
\hat{Q} = \frac{1}{n}\sum_{j=0}^{n-1} \big(\overline{U}\hat{U}\big)^j,
\end{align}
it is easy to check that the non-normalized $\ket{\psi'} =
\overline{UQ\ket{\phi}}+\hat{Q}\ket{\phi}$ solves our con-eigenvalue
problem. We can therefore replace $\ket{\phi}$ with
$\ket{\psi}=\ket{\psi'}/\sqrt{\braket{\psi'}{\psi'}}$ in
eq.~(\ref{eq:sicbound}) to search the set of con-eigenvectors of an
anti-unitary symmetry.  In particular, we take $\hat{U}=\hat{U}_{F_f'}$ to search
for fiducial vectors with the Fibonacci symmetry.

In practice the numerical search was performed by repeating a local
search from different initial trial vectors until a solution is found.
The local search used a C++ implementation of L-BFGS and only a few
trials were required to find a Fibonacci-Lucas fiducial vector in each
dimension up to $844$. Nonetheless, the search in dimension $2208$ proved
just beyond our reach.

By repeating the search for randomly chosen trial vectors under the
unitarily invariant Haar measure on $\C^d$, the entire search space
can be exhausted to identify all unique extended Clifford orbits
generated by Fibonacci-Lucas fiducial vectors. This was done for
dimensions up to $124$, where only a single Fibonacci-Lucas extended
Clifford orbit was found in each case. In all dimensions calculated
(all but $844$) the symmetry group of the orbit was found to be no
larger than that generated by $F_f$.  For dimension $844$, we used
the general approach\cite{GrWa17} based on triple-products
\begin{alignat}{5}
  t_{0ij}=\braket{\psi_0}{\psi_i}\braket{\psi_i}{\psi_j}\braket{\psi_j}{\psi_0}
=\tr(\ket{\psi_0}\bra{\psi_0}\cdot\ket{\psi_i}\bra{\psi_i}\cdot\ket{\psi_j}\bra{\psi_j})
=\tr(\Pi_0\Pi_i\Pi_j)
\end{alignat}
to verify that the additional unitary symmetry within the Clifford
group has order $21$. This implies that again the symmetry group is
exactly the one generated by $F_f$.

All numerical solutions are included as supplementary files in the
article source and also made available online.\cite{SICs_online} These
can be taken to be exact up to $150$ digits. Solutions in dimensions
up to $124$ match the exact solutions.  For dimension $323$, the
numerical solution has precisely the symmetry given by $F_f$, while
the symbolic solution is related by a Clifford transformation in order
to simplify the representation.
 
\section{Conclusions}

Identifying a putative family of SIC-POVMs for which the size of the
additional symmetry grows with the dimension allowed us to find both
exact and numerical solutions for the largest dimensions so far. For
dimension $844$, it doesn't seem to be completely out of reach to
obtain exact solutions.  The number of real parameters is still kind
of moderate in comparison to the dimension, we have to solve for at
most $42$ real variables. Based on the factorization $844=4\times 211$
and $211 \bmod 3 =1$, we can find a sparse representation of the
symmetry over a small cyclotomic field.  On the other hand, the degree
of the corresponding ray class field over the rationals is quite
large, namely $100\,800$.  When considered as extensions of the
corresponding cyclotomic field generated by a primitive $1688$-th
root of unity, however, the degree of the ray class field is only
$120$.  There is also some chance to convert the numerical solutions
into exact ones.\cite{ACFW17}

Similar to the Fibonacci matrix investigated here, there are other
candidates for matrices that give rise to putative families of
additional symmetries for Weyl-Heisenberg covariant SIC-POVMs.  We
leave this to future work.\cite{GrScSe17}

Despite the simplifications due to the additional symmetries, there
are clearly limits up to which dimension solutions can be found by
direct computation.  So we hope that the new families of putative
symmetries, in combination with all the structure of SIC-POVMs that
has been brought to light so far, will eventually give way to a 
construction of an infinite family of SIC-POVMs.

\acknowledgments The authors would like to thank Gary McConnell for
sharing his findings about the relations between SIC-POVMs in
different dimensions, as well as Ingemar Bengtsson for providing a
preliminary version of Ref. \citenum{ABDF17}.  Additional computing
resources provided by the Russel Division at MPL, Erlangen, as well as
Beno\^{\i}t Gr\'{e}maud at the Centre for Quantum Technologies,
Singapore, are acknowledged as well.


\begin{thebibliography}{17}%
\makeatletter
\providecommand \@ifxundefined [1]{%
 \@ifx{#1\undefined}
}%
\providecommand \@ifnum [1]{%
 \ifnum #1\expandafter \@firstoftwo
 \else \expandafter \@secondoftwo
 \fi
}%
\providecommand \@ifx [1]{%
 \ifx #1\expandafter \@firstoftwo
 \else \expandafter \@secondoftwo
 \fi
}%
\providecommand \natexlab [1]{#1}%
\providecommand \enquote  [1]{``#1''}%
\providecommand \bibnamefont  [1]{#1}%
\providecommand \bibfnamefont [1]{#1}%
\providecommand \citenamefont [1]{#1}%
\providecommand \href@noop [0]{\@secondoftwo}%
\providecommand \href [0]{\begingroup \@sanitize@url \@href}%
\providecommand \@href[1]{\@@startlink{#1}\@@href}%
\providecommand \@@href[1]{\endgroup#1\@@endlink}%
\providecommand \@sanitize@url [0]{\catcode `\\12\catcode `\$12\catcode
  `\&12\catcode `\#12\catcode `\^12\catcode `\_12\catcode `\%12\relax}%
\providecommand \@@startlink[1]{}%
\providecommand \@@endlink[0]{}%
\providecommand \url  [0]{\begingroup\@sanitize@url \@url }%
\providecommand \@url [1]{\endgroup\@href {#1}{\urlprefix }}%
\providecommand \urlprefix  [0]{URL }%
\providecommand \Eprint [0]{\href }%
\providecommand \doibase [0]{http://dx.doi.org/}%
\providecommand \selectlanguage [0]{\@gobble}%
\providecommand \bibinfo  [0]{\@secondoftwo}%
\providecommand \bibfield  [0]{\@secondoftwo}%
\providecommand \translation [1]{[#1]}%
\providecommand \BibitemOpen [0]{}%
\providecommand \bibitemStop [0]{}%
\providecommand \bibitemNoStop [0]{.\EOS\space}%
\providecommand \EOS [0]{\spacefactor3000\relax}%
\providecommand \BibitemShut  [1]{\csname bibitem#1\endcsname}%
\let\auto@bib@innerbib\@empty
\bibitem [{\citenamefont {Zauner}(1999)}]{Zau99}%
  \BibitemOpen
  \bibfield  {author} {\bibinfo {author} {\bibfnamefont {G.}~\bibnamefont
  {Zauner}},\ }\emph {\bibinfo {title} {Grundz{\"u}ge einer nichtkommutativen
  {D}esigntheorie}},\ \href@noop {} {Ph.D. thesis},\ \bibinfo  {school}
  {Universit{\"a}t Wien} (\bibinfo {year} {1999})\BibitemShut {NoStop}%
\bibitem [{\citenamefont {Zauner}(2011)}]{Zau11}%
  \BibitemOpen
  \bibfield  {author} {\bibinfo {author} {\bibfnamefont {G.}~\bibnamefont
  {Zauner}},\ }\bibfield  {title} {\enquote {\bibinfo {title} {Quantum designs:
  Foundations of a non-commutative design theory},}\ }\bibfield  {journal} {\bibinfo  {journal}
  {International Journal of Quantum Information}\ }\textbf {\bibinfo {volume}
  {9}},\ \href {\doibase
  10.1142/S0219749911006776} {\bibinfo {pages} {445--508} }(\bibinfo {year} {2011})\BibitemShut
  {NoStop}%
\bibitem [{\citenamefont {Renes}\ \emph {et~al.}(2004)\citenamefont {Renes},
  \citenamefont {Blume-Kohout}, \citenamefont {Scott},\ and\ \citenamefont
  {Caves}}]{RBSC04}%
  \BibitemOpen
  \bibfield  {author} {\bibinfo {author} {\bibfnamefont {J.~M.}\ \bibnamefont
  {Renes}}, \bibinfo {author} {\bibfnamefont {R.}~\bibnamefont {Blume-Kohout}},
  \bibinfo {author} {\bibfnamefont {A.~J.}\ \bibnamefont {Scott}}, \ and\
  \bibinfo {author} {\bibfnamefont {C.~M.}\ \bibnamefont {Caves}},\ }\bibfield
  {title} {\enquote {\bibinfo {title} {Symmetric informationally complete
  quantum measurements},}\ }\bibfield
  {journal} {\bibinfo  {journal} {Journal of Mathematical Physics}\ }\textbf
  {\bibinfo {volume} {45}},\ \href {\doibase 10.1063/1.1737053} {\bibinfo {pages} {2171}} (\bibinfo {year}
  {2004})\BibitemShut {NoStop}%
\bibitem [{\citenamefont {Scott}\ and\ \citenamefont {Grassl}(2010)}]{ScGr10}%
  \BibitemOpen
  \bibfield  {author} {\bibinfo {author} {\bibfnamefont {A.~J.}\ \bibnamefont
  {Scott}}\ and\ \bibinfo {author} {\bibfnamefont {M.}~\bibnamefont {Grassl}},\
  }\bibfield  {title} {\enquote {\bibinfo {title} {Symmetric informationally
  complete positive-operator-valued measures: {A} new computer study},}\ }\bibfield  {journal} {\bibinfo  {journal}
  {Journal of Mathematical Physics}\ }\textbf {\bibinfo {volume} {51}},\
  \href
  {\doibase 10.1063/1.3374022} {\bibinfo {pages} {042203}} (\bibinfo {year} {2010})\BibitemShut {NoStop}%
\bibitem [{\citenamefont {Scott}(2017)}]{Sco17}%
  \BibitemOpen
  \bibfield  {author} {\bibinfo {author} {\bibfnamefont {A.~J.}\ \bibnamefont
  {Scott}},\ }\enquote {\bibinfo
  {title} {{SICs:} extending the list of solutions},}\ {\bibinfo {howpublished}
  \href {https://arxiv.org/abs/1703.03993}{arXiv:1703.03993} [quant-ph]} (\bibinfo {year} {2017})\BibitemShut {NoStop}%
\bibitem [{\citenamefont {Fuchs}, \citenamefont {Hoang},\ and\ \citenamefont
  {Stacey}(2017)}]{FHS17}%
  \BibitemOpen
  \bibfield  {author} {\bibinfo {author} {\bibfnamefont {C.~A.}\ \bibnamefont
  {Fuchs}}, \bibinfo {author} {\bibfnamefont {M.~C.}\ \bibnamefont {Hoang}}, \
  and\ \bibinfo {author} {\bibfnamefont {B.~C.}\ \bibnamefont {Stacey}},\
  }\bibfield  {title} {\enquote {\bibinfo {title} {The {SIC} question: History
  and state of play},}\ } {\bibfield
  {journal} {\bibinfo  {journal} {axioms}\ }\textbf {\bibinfo {volume} {6}},\
  \bibinfo {pages} {\href {\doibase 10.3390/axioms6030021}{21}} (\bibinfo {year} {2017})}\BibitemShut {NoStop}%
\bibitem [{\citenamefont {Appleby}(2005)}]{App05}%
  \BibitemOpen
  \bibfield  {author} {\bibinfo {author} {\bibfnamefont {D.~M.}\ \bibnamefont
  {Appleby}},\ }\bibfield  {title} {\enquote {\bibinfo {title} {Symmetric
  informationally complete-positive operator valued measures and the extended
  {Clifford} group},}\ }\bibfield
  {journal} {\bibinfo  {journal} {Journal of Mathematical Physics}\ }\textbf
  {\bibinfo {volume} {46}},\ \href {\doibase 10.1063/1.1896384} {\bibinfo {pages} {052107}} (\bibinfo {year}
  {2005})\BibitemShut {NoStop}%
\bibitem [{\citenamefont {Grassl}(2005)}]{Gra05}%
  \BibitemOpen
  \bibfield  {author} {\bibinfo {author} {\bibfnamefont {M.}~\bibnamefont
  {Grassl}},\ }\bibfield  {title} {\enquote {\bibinfo {title} {Tomography of
  quantum states in small dimensions},}\ }\bibfield  {journal} {\bibinfo  {journal}
  {Electronic Notes in Discrete Mathematics}\ }\textbf {\bibinfo {volume}
  {20}},\ \href {\doibase
  10.1016/j.endm.2005.05.060} {\bibinfo {pages} {151--164}} (\bibinfo {year} {2005}),\ \bibinfo
  {note} {{W}orkshop on Discrete Tomography and its Applications New York, USA,
  June 2005}\BibitemShut {NoStop}%
\bibitem [{\citenamefont {Grassl}\ and\ \citenamefont
  {Waldron}(2017)}]{GrWa17}%
  \BibitemOpen
  \bibfield  {author} {\bibinfo {author} {\bibfnamefont {M.}~\bibnamefont
  {Grassl}}\ and\ \bibinfo {author} {\bibfnamefont {S.}~\bibnamefont
  {Waldron}},\ }\href@noop {} {\enquote {\bibinfo {title} {Computing projective
  symmetries of frames},}\ } (\bibinfo {year} {2017}),\ \bibinfo {note} {in
  preparation}\BibitemShut {NoStop}%
\bibitem [{\citenamefont {Sloane}(2010)}]{OEIS}%
  \BibitemOpen
  \bibfield  {author} {\bibinfo {author} {\bibfnamefont {N.~J.~A.}\
  \bibnamefont {Sloane} (editor)},\ }\enquote {\bibinfo
  {title} {The on-line encyclopedia of integer sequences},}\ \bibinfo
  {howpublished} {published electronically at \href {https://oeis.org} {https://oeis.org}} (\bibinfo
  {year} {2010})\BibitemShut {NoStop}%
\bibitem [{\citenamefont {Graham}, \citenamefont {Knuth},\ and\ \citenamefont
  {Patashnik}(1994)}]{GKP94}%
  \BibitemOpen
  \bibfield  {author} {\bibinfo {author} {\bibfnamefont {R.~L.}\ \bibnamefont
  {Graham}}, \bibinfo {author} {\bibfnamefont {D.~E.}\ \bibnamefont {Knuth}}, \
  and\ \bibinfo {author} {\bibfnamefont {O.}~\bibnamefont {Patashnik}},\
  }\href@noop {} {\emph {\bibinfo {title} {Concrete Mathematics: A Foundation
  for Computer Science}}},\ \bibinfo {edition} {2nd}\ ed.\ (\bibinfo
  {publisher} {Addison-Wesley},\ \bibinfo {address} {Boston, MA, USA},\
  \bibinfo {year} {1994})\BibitemShut {NoStop}%
\bibitem [{\citenamefont {Appleby}\ \emph {et~al.}(2016)\citenamefont
  {Appleby}, \citenamefont {Flammia}, \citenamefont {McConnell},\ and\
  \citenamefont {Yard}}]{AFMY16}%
  \BibitemOpen
  \bibfield  {author} {\bibinfo {author} {\bibfnamefont {M.}~\bibnamefont
  {Appleby}}, \bibinfo {author} {\bibfnamefont {S.}~\bibnamefont {Flammia}},
  \bibinfo {author} {\bibfnamefont {G.}~\bibnamefont {McConnell}}, \ and\
  \bibinfo {author} {\bibfnamefont {J.}~\bibnamefont {Yard}},\ }\enquote {\bibinfo {title} {Generating
  ray class fields of real quadratic fields via complex equiangular lines},}\
   {\bibinfo {howpublished} {\href
  {https://arxiv.org/abs/1604.06098}{arXiv:1604.06098} [math.NT]}} (\bibinfo
  {year} {2016})\BibitemShut {NoStop}%
\bibitem [{SIC(2017)}]{SICs_online}%
  \BibitemOpen
  \href {http://sicpovm.markus-grassl.de} { {\bibinfo {title}
  {http://sicpovm.markus-grassl.de},}\ } (\bibinfo {year} {2017})\BibitemShut
  {NoStop}%
\bibitem [{\citenamefont {Appleby}\ \emph {et~al.}(2014)\citenamefont
  {Appleby}, \citenamefont {Bengtsson}, \citenamefont {Brierley}, \citenamefont
  {Ericsson}, \citenamefont {Grassl},\ and\ \citenamefont
  {Larsson}}]{ABBEGL04}%
  \BibitemOpen
  \bibfield  {author} {\bibinfo {author} {\bibfnamefont {D.~M.}\ \bibnamefont
  {Appleby}}, \bibinfo {author} {\bibfnamefont {I.}~\bibnamefont {Bengtsson}},
  \bibinfo {author} {\bibfnamefont {S.}~\bibnamefont {Brierley}}, \bibinfo
  {author} {\bibfnamefont {{\AA}.}~\bibnamefont {Ericsson}}, \bibinfo {author}
  {\bibfnamefont {M.}~\bibnamefont {Grassl}}, \ and\ \bibinfo {author}
  {\bibfnamefont {J.-{\AA}.}\ \bibnamefont {Larsson}},\ }\bibfield  {title}
  {\enquote {\bibinfo {title} {Systems of imprimitivity for the {Clifford}
  group},}\ }\href@noop {} {\bibfield  {journal} {\bibinfo  {journal} {Quantum
  Information \& Computation}\ }\textbf {\bibinfo {volume} {14}},\ \bibinfo
  {pages} {339--360} (\bibinfo {year} {2014})},\ {\bibinfo {note} {
  \href
  {https://arxiv.org/abs/1210.1055}{arXiv:1210.1055} [quant-ph]}}\BibitemShut {NoStop}%
\bibitem [{\citenamefont {Bosma}, \citenamefont {Cannon},\ and\ \citenamefont
  {Playoust}(1997)}]{Magma}%
  \BibitemOpen
  \bibfield  {author} {\bibinfo {author} {\bibfnamefont {W.}~\bibnamefont
  {Bosma}}, \bibinfo {author} {\bibfnamefont {J.~J.}\ \bibnamefont {Cannon}}, \
  and\ \bibinfo {author} {\bibfnamefont {C.}~\bibnamefont {Playoust}},\
  }\bibfield  {title} {\enquote {\bibinfo {title} {{The Magma Algebra System I:
  The User Language}},}\ }\href@noop {} {\bibfield  {journal} {\bibinfo
  {journal} {Journal of Symbolic Computation}\ }\textbf {\bibinfo {volume}
  {24}},\ \bibinfo {pages} {235--265} (\bibinfo {year} {1997})}\BibitemShut
  {NoStop}%
\bibitem [{\citenamefont {Appleby}\ \emph
  {et~al.}(2017{\natexlab{a}})\citenamefont {Appleby}, \citenamefont
  {Bengtsson}, \citenamefont {Dumitru},\ and\ \citenamefont
  {Flammia}}]{ABDF17}%
  \BibitemOpen
  \bibfield  {author} {\bibinfo {author} {\bibfnamefont {M.}~\bibnamefont
  {Appleby}}, \bibinfo {author} {\bibfnamefont {I.}~\bibnamefont {Bengtsson}},
  \bibinfo {author} {\bibfnamefont {I.}~\bibnamefont {Dumitru}}, \ and\
  \bibinfo {author} {\bibfnamefont {S.}~\bibnamefont {Flammia}},\ }{\enquote {\bibinfo {title} {Dimension
  towers of {SICs}. {I}. {Aligned} {SICs} and embedded tight frames},}\
  }\bibinfo {howpublished} {pre\-print \href
  {https://arxiv.org/abs/1707.09911} {arXiv:1707.09911} [quant-ph]} (\bibinfo
  {year} {2017}{\natexlab{a}})\BibitemShut {NoStop}%
\bibitem [{\citenamefont {Appleby}\ \emph {et~al.}(2017)\citenamefont
  {Appleby}, \citenamefont {Chien}, \citenamefont {Flammia},\ and\
  \citenamefont {Waldron}}]{ACFW17}%
  \BibitemOpen
  \bibfield  {author} {\bibinfo {author} {\bibfnamefont {M.}~\bibnamefont
  {Appleby}}, \bibinfo {author} {\bibfnamefont {T.-Y.}\ \bibnamefont {Chien}},
  \bibinfo {author} {\bibfnamefont {S.}~\bibnamefont {Flammia}}, \ and\
  \bibinfo {author} {\bibfnamefont {S.}~\bibnamefont {Waldron}},\ }\enquote {\bibinfo {title} {Constructing
  exact symmetric informationally complete measurements from numerical
  solutions},}\  {\bibinfo {howpublished} {\href
  {https://arxiv.org/abs/1703.05981}{arXiv:1703.05981}
  [quant-ph]}} (\bibinfo {year} {2017})\BibitemShut {NoStop}%
\bibitem [{\citenamefont {Grassl}, \citenamefont {Scott},\ and\ \citenamefont
  {Seyfarth}(2017)}]{GrScSe17}%
  \BibitemOpen
  \bibfield  {author} {\bibinfo {author} {\bibfnamefont {M.}~\bibnamefont
  {Grassl}}, \bibinfo {author} {\bibfnamefont {A.~J.}\ \bibnamefont {Scott}}, \
  and\ \bibinfo {author} {\bibfnamefont {U.}~\bibnamefont {Seyfarth}},\
  }\href@noop {} {\enquote {\bibinfo {title} {Symmetries of {W}eyl-{H}eisenberg
  {SIC-POVMs}},}\ } (\bibinfo {year} {2017}),\ \bibinfo {note} {in
  preparation}\BibitemShut {NoStop}%
\bibitem [{\citenamefont {Euler}(1765)}]{Eul1765}%
  \BibitemOpen
  \bibfield  {author} {\bibinfo {author} {\bibfnamefont {L.}~\bibnamefont
  {Euler}},\ }\bibfield  {title} {\enquote {\bibinfo {title} {Observationes
  {analytic\ae}},}\ }\href@noop {} {\bibfield  {journal} {\bibinfo  {journal}
  {Novi commentarii {academi\ae} scientiarum {Petropolitan\ae}},\ \bibinfo
  {pages} {124--143}} (\bibinfo {year} {1765})},\ \bibinfo {note} {reprinted in
  his Opera Omnia, series 1, volume 15, 50--69}\BibitemShut {NoStop}%
\end{thebibliography}
%

\appendix
\section{Relations for Fibonacci and Lucas Numbers}
In the following we will derive a couple of relations for Fibonacci
and Lucas numbers.

\begin{proposition}
  Let $\varphi=\frac{1+\sqrt{5}}{2}$ be the golden ratio which is the
  positive root of the equation $x^2=x+1$. Then we have
  the following well-known closed formulas for the Fibonacci
  \cite{Eul1765} and Lucas
  numbers (see, e.\,g., Exercise~28 in Ref.~\citenum{GKP94}):
  \begin{alignat}{5}
    F_n&{}=\frac{\varphi^n-(1-\varphi)^n}{\sqrt{5}}
      &&{}=\frac{\varphi^n-(-\varphi)^{-n}}{\sqrt{5}}\label{eq:fibobacci_closed}\\
    \text{and}\quad
    L_n&{}=\varphi^n+(1-\varphi)^n
      &&{}=\varphi^n+(-\varphi)^{-n}.\label{eq:lucas_closed}
  \end{alignat}
\end{proposition}

\begin{proposition}
  For $k,n\ge 1$, the following identities hold:
  \begin{alignat}{5}
    i)&\qquad&F_n&{}\mid F_{kn}\label{eq:divisibility_Fibonacci}\\
    ii)&&    L_n&{}=F_{n-1}+F_{n+1}\label{eq:lucas_sum_fibonacci}\\
    iii)&&      L_{2n}^2&{}=5
    F_{2n}^2+4\label{eq:squarefree}\\
    iv)&&    F_{6n}&{}=F_{2n} (L_{4n}+1)\label{eq:F_6n}\\
    v)&&    L_{4n}+1&{}=(L_{2n}+1)(L_{2n}-1)=L_{2n}^2-1\label{eq:L_4n}\\
    vi)&& F_{6n-1}+L_{2n}&{}= F_{2n-1}(L_{4n}+1)\label{eq:F_6n-1a}
  \end{alignat}
\end{proposition}
\begin{proof}
  \begin{description}
    \item[\normalfont$i)$] By induction it can be shown that $F_{kn}$
      is divisible by $F_n$ (see, e.\,g., Eq.~(6.111) in
      Ref.~\citenum{GKP94}).
    \item[\normalfont$ii)$] We prove the identity by induction.  It
      clearly holds for $n=1$. Using the defining
      relations, we compute
      \begin{alignat}{5}
        L_{n+1}=L_n+L_{n-1}&{}=(F_{n-1}+F_{n+1})+(F_{n-2}+F_{n})\nonumber\\
        &{}=F_{n-1}+F_{n-2}+F_{n+1}+F_{n}
        &{}=F_n+F_{n+2}.
      \end{alignat}
    \item[\normalfont$iii)$] Using the closed formulas \eqref{eq:fibobacci_closed}
      and \eqref{eq:lucas_closed}, we compute
      \begin{alignat}{5}
      5 F_{2n}^2+4{}& =(\varphi^{2n}-(-\varphi)^{-2n})^2+4 
      &&{}=\varphi^{4n}-2+(-\varphi)^{-4n}+4\nonumber\\
      &&&{}=\varphi^{4n}+2+(-\varphi)^{-4n}\nonumber\\
      &&&{}=(\varphi^{2n}+(-\varphi)^{-2n})^2
      &{}=L_{2n}^2.
      \end{alignat}
    \item[\normalfont$iv)$] Similar as before, we compute
      \begin{alignat}{9}
        F_{2n}(L_{4n}+1)&{}=\frac{1}{\sqrt{5}}(\varphi^{2n}-(-\varphi)^{-2n})(\varphi^{4n}+(-\varphi)^{-4n}+1)\nonumber\\
        &{}=\frac{1}{\sqrt{5}}(\varphi^{6n}-(-\varphi)^{-6n})
        &{}=F_{6n}.
      \end{alignat}
    \item[\normalfont$v)$] Again, direct computations shows
      \begin{alignat}{9}
        (L_{2n}+1)(L_{2n}-1)&{}=(\varphi^{2n}+(-\varphi)^{-2n}+1)(\varphi^{2n}+(-\varphi)^{-2n}-1)\nonumber\\
        &{}=\varphi^{4n}+(-\varphi)^{-4n}+1
        &{}=L_{4n}+1.
      \end{alignat}
    \item[\normalfont$vi)$] In order to prove the claim, consider
      \begin{alignat}{9}
        F_{2n-1}(L_{4n}+1)
        &{}=\frac{1}{\sqrt{5}}(\varphi^{2n-1}-(-\varphi)^{-2n+1})(\varphi^{4n}+(-\varphi)^{-4n}+1)\nonumber\\
        &{}=\frac{1}{\sqrt{5}}(\varphi^{6n-1}+\varphi^{2n+1}+\varphi^{2n-1}+\varphi^{-2n+1}+\varphi^{-2n-1}+\varphi^{-6n+1})\nonumber\\
        &{}=\frac{\varphi^{6n-1}+\varphi^{-6n+1}}{\sqrt{5}}+\frac{(\varphi^{2n}+\varphi^{-2n})(\varphi+\varphi^{-1})}{\sqrt{5}}\nonumber\\
        &{}=\frac{\varphi^{6n-1}-(-\varphi)^{-6n+1}}{\sqrt{5}}+\varphi^{2n}+(-\varphi)^{-2n}\nonumber\\
        &{}=F_{6n-1}+L_{2n}
      \end{alignat}
  \end{description}
\end{proof}

\end{document}